\documentclass[a4paper,10pt]{amsart}

\usepackage{amsmath,amssymb,amsthm,a4wide}

\usepackage{tikz}
\usetikzlibrary{shapes}
\usetikzlibrary{intersections}

\usepackage{graphicx}
\usepackage{graphics}

\newtheorem*{thm}{Theorem}

\begin{document}

\title[Filtering of Markov chains and spectral embedding]{A filtering technique for Markov chains \\
with applications to spectral embedding}
\author{Stefan Steinerberger}
\thanks{Stefan Steinerberger, Department of Mathematics, Yale University, 10 Hillhouse Avenue, New Haven, CT 06511, USA, e-mail: \textsc{stefan.steinerberger@yale.edu}, phone number: 475-202-2136}

\begin{abstract} 
Spectral methods have proven to be a highly effective tool in understanding the intrinsic geometry of a high-dimensional data set $\left\{x_i \right\}_{i=1}^{n} \subset \mathbb{R}^d$. The
key ingredient is the construction of a Markov chain on the set, where transition probabilities depend on the distance between elements, for example where for every $1 \leq j \leq n$ the 
probability of going from
$x_j$ to $x_i$ is proportional to
$$ p_{ij} \sim \exp \left( -\frac{1}{\varepsilon}\|x_i -x_j\|^2_{\ell^2(\mathbb{R}^d)}\right) \qquad \mbox{where}~\varepsilon>0~\mbox{is a free parameter}.$$
We propose a method which increases the self-consistency of such Markov chains before spectral methods are applied. 
Instead of directly using a Markov transition matrix $P$, we set $p_{ii} = 0$ and rescale, thereby obtaining a transition matrix $P^*$ modeling a non-lazy random walk.
We then create a new transition matrix $Q = (q_{ij})_{i,j=1}^{n}$ by demanding that for fixed $j$ the quantity $q_{ij}$ be proportional to
$$ q_{ij} \sim \min((P^*)_{ij}, ((P^*)^2)_{ij}, \dots, ((P^*)^k)_{ij}) \qquad \mbox{where usually}~ k=2.$$
We consider several classical data sets, show that this simple method can increase the efficiency of spectral methods and prove that it can correct randomly introduced errors in the kernel.
\end{abstract}
\maketitle

\section{Introduction}
We study the interplay between Markov chains on a high-dimensional data set $\left\{x_i\right\}_{i=1}^{n} \subset \mathbb{R}^d$ and the inner workings of spectral methods. There are many different methods, see e.g. the work of Belkin \& Niyogi \cite{belk}, Coifman \& Lafon \cite{coif1}, Coifman \& Maggioni \cite{coif2}, Donoho \& Grimes \cite{donoho},
Roweis \& Saul \cite{rs} and  Tenenbaum, de Silva \& Langford \cite{ten}. We propose a natural operation on the data which will enhance the effectivity of these methods. Usually,
these techniques proceed by imposing a Markov chain on the data set and analyzing diffusion on the arising graph. A popular and natural choice for the Markov chain is to
declare that the probability $p_{ij}$ to move from point $x_j$ to $x_i$ is
$$ p_{ij} =  \frac{ \exp\left(-\frac{1}{\varepsilon}\|x_i - x_j\|^2_{\ell^2(\mathbb{R}^d)}\right)}{\sum_{k=1}^{n}{ \exp\left(-\frac{1}{\varepsilon}\|x_k - x_j\|^2_{\ell^2(\mathbb{R}^d)}\right)}},$$
where the value of $\varepsilon$ needs to be chosen depending on the given data as it induces a natural length scale $\sim \sqrt{\varepsilon}$ which should match the distance between neighbouring points. 
If the number of
random jumps tends to infinity, the probability of being in a particular point is uniformly distributed and contains no more information on the geometry: the stationary solution
$$ P\pi = \pi \quad \mbox{has only the trivial solution} \quad \pi = (1,1,1,\dots, 1)^t.$$ 

It is not difficult to see that the initial states converging to the stationary state as slowly as possible are of particular intrinsic interest because they are exploiting geometric
bottlenecks in the data to avoid rapid mixing.
Starting from the seminal work of Cheeger \cite{ch}, it is now
understood that the first (nontrivial) eigenfunction of the Laplacian
will change its sign at a region of geometrical significance (see Figure 1).
\begin{figure}[h!]
\includegraphics[width=6cm]{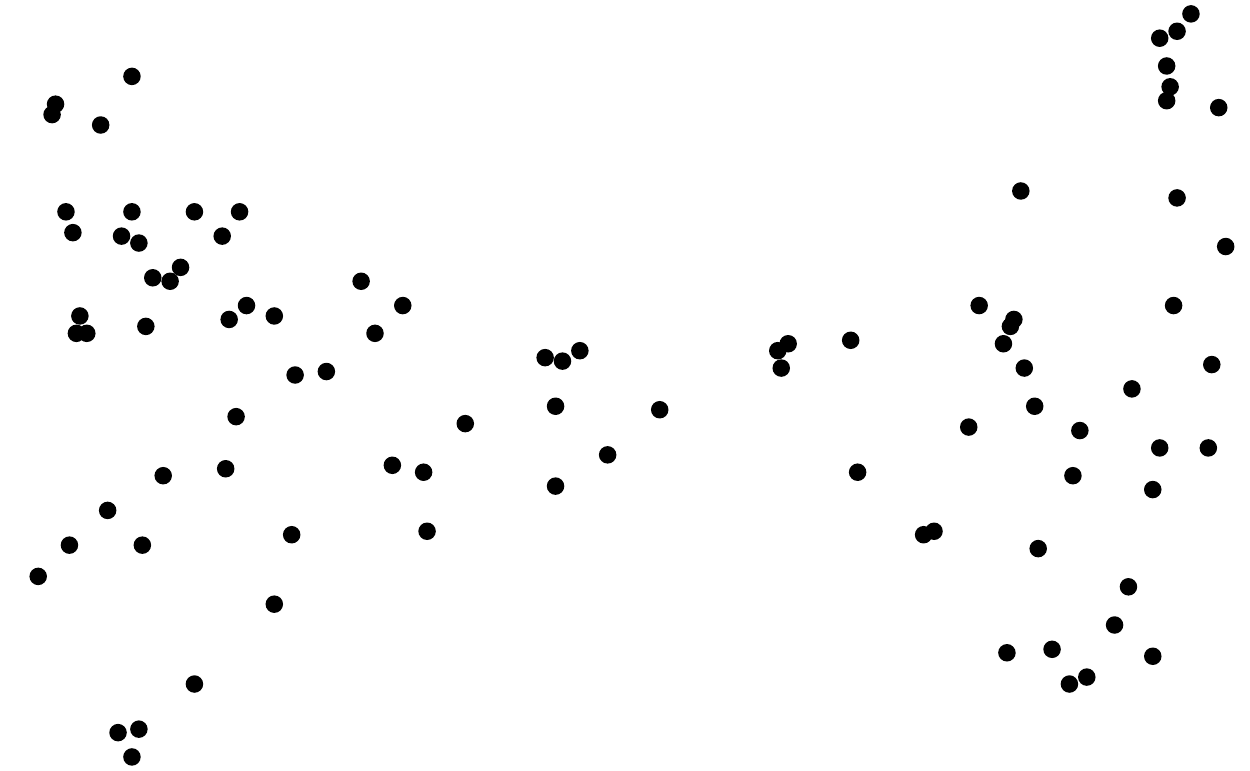}
\caption{Two regions of points connected by a bottleneck. The diffusion distance respects the bottleneck (a random walk started in the left part is more likely, at least for short time, to remain there than to cross the bottleneck) while Euclidean distance does not -- furthermore, the first eigenfunction will switch its sign in the bottleneck and allow for an identification of both regions.}
\end{figure}
Given $k$ eigenfunctions on the graph, we may use them to construct the embedding
\begin{align*}
\Phi: \left\{x_i\right\}_{i=1}^{n} &\rightarrow \mathbb{R}^k \\
x &\rightarrow  (\phi_1(x), \phi_2(x), \dots, \phi_k(x)).
\end{align*}

At this level of generality, it is not possible to make any qualitative statements about the quality of the embedding depending on the data set and $k$. However, it has been shown that in
practice already a very small number of eigenfunction is able to reliably extract relevant geometric features from the data set.

\section{Rewarding self-consistency in the Markov chain}
\subsection{Motivation.} We are assuming that the point set $\left\{x_i\right\}_{i=1}^{n}$ comes from a (possibly very low-dimensional) manifold embedded in $\mathbb{R}^d$. The following example
is representative: assume we are given a point $x_1$ embedded in a point cloud given in Figure 2. 
\begin{center}
\begin{figure}[h!]
\begin{tikzpicture}[scale=1.3]
\draw [fill, ultra thick] (0,1) circle [radius=0.03];
\draw [fill, ultra thick] (1,0) circle [radius=0.03];
\draw [fill, ultra thick] (1,1) circle [radius=0.03];
\draw [fill, ultra thick] (1,2) circle [radius=0.03];

\draw [fill, ultra thick] (-1,1) circle [radius=0.03];
  \draw[ thick] (0,1) --(-1, 1);

  \draw[ thick] (0,1) --(1, 0);
  \draw[ thick] (0,1) --(1, 1);
  \draw[ thick] (0,1) --(1, 2);
  \draw[ thick] (1,1) --(1, 2);
  \draw[ thick] (1,1) --(1, 0);
  \draw[ thick] (1,1) --(1, 0);
 \node[right] at (-0.3,0.7) {$x_1$};
  \draw [ thick] (1,2) to[out=325,in=35] (1,0);
 \node[right] at (-1.3,0.7) {$x_5$};

 \node[right] at (1.1,2.1) {$x_2$};
 \node[right] at (1.03,1) {$x_3$};
 \node[right] at (1.1,-0.1) {$x_4$};
\end{tikzpicture}
\caption{A graph representing local structure around $x_1$.}
\end{figure}
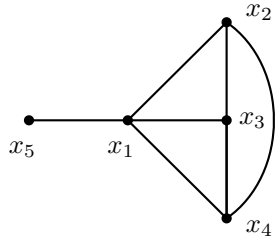
\end{center}
\vspace{-15pt}
Given such a local structure, one would certainly assume that $x_1, x_2, x_3, x_4$ belong to points on the manifold, which are close to each other; $x_5$ is close to $x_1$ but not close
to any of the other points. It seems that, whatever the hidden manifold $\mathcal{M} \subset \mathbb{R}^d$ may be, the guess $x_1,x_2, x_3, x_4 \in \mathcal{M}$ is much more reasonable than
$x_1, x_5 \in \mathcal{M}$. As it turns out, there is a very simple trick that allows to detect such structures.
Consider the probability of being in one of the five points after 1 and 2 steps of a random walk starting in $x_1$.
\begin{center}
\begin{figure}[h!]
\begin{tabular}{ l c c c c r }
  probability of being in  & $x_1$ & $x_2$ & $x_3$ & $x_4$ & $x_5$ \\
after 1 step & 0 & 1/4 & 1/4 & 1/4 & 1/4\\
after 2 steps & 1/2 & 1/6 & 1/6 & 1/6 & 0\\
minimum & 0 & 1/6 & 1/6 & 1/6 & 0
\end{tabular}
\caption{The probability of being in a point after 1 and 2 steps: it automatically selects the 'likely' neighbours $x_2, x_3, x_4$.}
\end{figure}
\end{center}
\vspace{-19pt}
The minimum can be interpreted as a measure of how easy it is to reach a certain point in \textit{both} one and two steps starting from $x_1$. If an adjacent point is relatively easy to reach in 
one step but impossible or very difficult to reach in two, then that point may be as close or even closer to $x_1$ as other points but it will not be close to those other points. It should be emphasized
that for this reasoning to be valid we need to ensure that
\begin{itemize}
\item the random walk is not lazy: it must not be possible to remain at a position and
\item that the random walk is interconnected enough.
\end{itemize}
For Markov chains created using the isotropic kernels with a Gaussian weight, the second condition is always satisfied and the first one can be readily satisfied by setting $p_{ii} = 0$ for all
$1 \leq i \leq n$ (and rescaling the other elements in the column).

\subsection{The method} Formally, our algorithm may be described as follows. Assume we are given $\left\{x_i\right\}_{i=1}^{n} \subset \mathbb{R}^d$ and an associated Markov chain described
by the matrix $P = (p_{ij})_{i,j=1}^{n}$. The matrix $Q$ will be created as follows: we obtain the matrix $P^*$ by setting $p_{ii} = 0$ and rescaling every column of $P$ so that to we are once again given 
a transition matrix. $Q$ is then given by
$$  \boxed{   q_{ij} = \frac{\min((P^*)_{ij}, ((P^*)^2)_{ij}, \dots, ((P^*)^k)_{ij}),}{\sum_{m=1}^{n}{\min((P^*)_{mj}, ((P^*)^2)_{mj}, \dots, ((P^*)^k)_{mj})}}.  }$$
We may then proceed with an analysis of the data set using $Q$ instead of $P$. It seems that $k=2$ is most effective in practice but there are certainly
cases where a larger $k$ may prove advantageous.

\section{Error detection.} The method provides a tool
to identify erroneous edges in graph structures satisfying the natural analogue of a manifold assumption for the point cloud: no point has too many neighbours and neighbours of neighbours of neighbours tend to be neighbours or at least close. It needs to be emphasized that the method relies strongly on the structure/kernel used in spectral embedding.
In the (unweighted) example below, the method would remove the dashed line, which may or may not be reasonable -- for the relevant applications
we have in mind, the method will always start from a \textit{complete, weighted} graph and the method will assign new weights to the edges.

\begin{center}
\begin{figure}[h!]
\begin{tikzpicture}[scale=1.3]
\draw [fill, ultra thick] (0,1) circle [radius=0.03];
\draw [fill, ultra thick] (1,0) circle [radius=0.03];
\draw [fill, ultra thick] (1,1) circle [radius=0.03];
\draw [fill, ultra thick] (1,2) circle [radius=0.03];

\draw [fill, ultra thick] (-1,1) circle [radius=0.03];
\draw [fill, ultra thick] (-2,2) circle [radius=0.03];
\draw [fill, ultra thick] (-2,0) circle [radius=0.03];
\draw [ thick] (-2, 0) -- (-2,2) -- (-1,1) -- (-2,0);
  \draw[dashed, thick] (0,1) --(-1, 1);

  \draw[ thick] (0,1) --(1, 0);
  \draw[ thick] (0,1) --(1, 1);
  \draw[ thick] (0,1) --(1, 2);
  \draw[ thick] (1,1) --(1, 2);
  \draw[ thick] (1,1) --(1, 0);
  \draw[ thick] (1,1) --(1, 0);

  \draw [ thick] (1,2) to[out=325,in=35] (1,0);

\end{tikzpicture}
\caption{$Q$ will remove the dashed line.}
\end{figure}
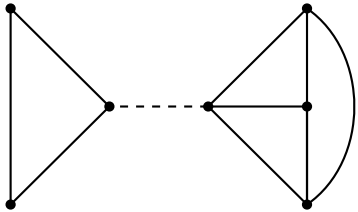
\end{center}

\subsection{Earlier results.} 
There is a lot of work on the fundamental problem of analyzing graphs with respect to erroneous edges (see e.g. \cite{ref1,ref2}). We emphasize that our approach 
should not been seen from that perspective because it crucially relies on the particular structure arising from the analysis of point clouds in $\mathbb{R}^n$ by means
of a nonvanishing kernel. It will, without further adaptation, not be useful in the more general graph setting: for example, it is easy to see that on the lattice $\mathbb{Z}^2$ 
with the natural edge relation
$$ (a,b) \sim (c,d) \qquad \mbox{if and only if} ~ |a-c| + |b-d| = 1$$
the method will remove all edges. However, it is certainly conceivable that there are natural ways (for instance weighting the edges proportional to $\exp{(-\mbox{dist}(x,y))}$), the method might 
also be useful in that context. As for the special problem of spectral embeddings, there is paper of David \& Averbuch \cite{av} whose approach is very different from ours.

\subsection{Setup.} The purpose of this section is a simple argument quantifying the ability of our method to detect and correct errors -- an explicit example is given in the next section.
The model is as follows: let $G = (V, E)$ be an unweighted graph on $|V| = n$ vertices such that for every $v \in V$ the number of points at distance $\leq 2$ from $v$ is at most $c \ll n$.
 let $P_1$ be the transition matrix of the (non-lazy) random walk on $G$. 
Now we create a random perturbation $P_2$ of $P_1$ by randomly adding edges. Let $0 \leq p < 1$ be a fixed
parameter and create $P_2$ by additionally connecting every pair of vertices $(v_1, v_2) \in V\times V$ with probability $p$ and taking $P_2$ to be the transition matrix of the non-lazy random 
walk on the enlarged graph. Let $Q$ be the modified Markov chain obtained from $P_2$ using our method with $k=2$. We note that $Q$ may separate connections which have
not been randomly added (see the example of the lattice), however, it is certainly very successful at removing those that have been added. When using a Gaussian kernel on
point clouds in applications, the triangle inequality will already guarantee that such structures cannot arise and one can expect $Q$ to rarely mistake an edge for an error. \\

Our main statement gives an upper bound on the expected number of random edges that have been added and are not being removed by $Q$.
Setting $p \sim \varepsilon/n$ for $\varepsilon$ small, it implies that up to $\sim \varepsilon n$ random edges will be completely filtered out with high likelihood $\geq 1- c \varepsilon$.
Current techniques \cite{ref1,ref2} for the generalized problem on graphs seem to be effective up to the same order (i.e. work for a number of randomly added edges that is linear in the number
of vertices).

\begin{thm} The number $N$ of vertices $(v_i, v_j) \in V \times V$ with
$$ (P_1)_{i,j} = 0 \qquad \mbox{and} \qquad Q_{ij} > 0$$
satisfies
$$ \mathbb{E}N \leq c n p + c n^2 p^2 + \frac{1}{2}n^3 p^3.$$
\end{thm}
\begin{proof} We start by analyzing the implications of
$$ (P_1)_{i,j} = 0 \qquad \mbox{and} \qquad Q_{ij} > 0.$$
Let us abbreviate $x = v_i$ and $y = v_j$. $(P_1)_{i,j} = 0$ implies that $x$ and $y$ are not connected by an edge. $Q_{ij} > 0$ implies that in $P_2$ there is
both an edge between $x$ and $y$ as well as an edge $z$ with $x \sim z$ and $z \sim y$. The edge $x \sim y$ must be a randomly added edge but the
status of the edges $x \sim z$ and $z \sim y$ cannot be deduced. We proceed with a case distinction depending on which edges have been there originally.\\

\begin{figure}[!htb]
\minipage{0.25\textwidth}
\begin{center}
\begin{tikzpicture}[scale=1.3]
\draw [fill, ultra thick] (0,0) circle [radius=0.03];
\draw [fill, ultra thick] (1,0) circle [radius=0.03];
\draw [fill, ultra thick] (0.5,0.866) circle [radius=0.03];
 \node[below] at (0.5,-0.2) {Case 1};
 \node[right] at (1.1,0) {$y$};
 \node[left] at (-0.1,0) {$x$};
 \node[above] at (1/2,0.87) {$z$};
\draw [dashed] (0,0) -- (1,0) -- (0.5, 0.866) -- (0,0);
\end{tikzpicture}
\end{center}
\endminipage\hfill
\minipage{0.25\textwidth}
\begin{center}
\begin{tikzpicture}[scale=1.3]
\draw [fill, ultra thick] (0,0) circle [radius=0.03];
\draw [fill, ultra thick] (1,0) circle [radius=0.03];
\draw [fill, ultra thick] (0.5,0.866) circle [radius=0.03];
 \node[below] at (0.5,-0.2) {Case 2};
 \node[right] at (1.1,0) {$y$};
 \node[left] at (-0.1,0) {$x$};
 \node[above] at (1/2,0.87) {$z$};
\draw [dashed] (0,0) -- (1,0);
\draw  [thick] (1,0) -- (0.5, 0.866) -- (0,0);
\end{tikzpicture}
\end{center}
\endminipage\hfill
\minipage{0.49\textwidth}%
\begin{center}
\begin{tikzpicture}[scale=1.3]
\draw [fill, ultra thick] (0,0) circle [radius=0.03];
\draw [fill, ultra thick] (1,0) circle [radius=0.03];
\draw [fill, ultra thick] (0.5,0.866) circle [radius=0.03];
 \node[right] at (1.1,0) {$y$};
 \node[left] at (-0.1,0) {$x$};
 \node[above] at (1/2,0.87) {$z$};
\draw [dashed] (0,0) -- (1,0) -- (0.5, 0.866);
\draw  [thick]  (0.5, 0.866) -- (0,0);
 \node[below] at (1.75,-0.2) {Case 3};
\draw [fill, ultra thick] (2.5,0) circle [radius=0.03];
\draw [fill, ultra thick] (3.5,0) circle [radius=0.03];
\draw [fill, ultra thick] (3,0.866) circle [radius=0.03];
 \node[right] at (3.6,0) {$y$};
 \node[left] at (2.5-0.1,0) {$x$};
 \node[above] at (2.5+1/2,0.87) {$z$};
\draw [dashed] (3, 0.866) -- (2.5,0) -- (3.5, 0);
\draw  [thick] (3.5,0) -- (3, 0.866);

\end{tikzpicture}
\end{center}
\endminipage
\end{figure}

\textit{Case 1.} All three edges $x \sim y$, $x \sim z$ and $z \sim y$ have been added by the random process. In this case adding the random edges has given rise to a
triangle. It is easily seen by linearity of expectation that the expected number of triangles is given by
$$ \binom{n}{3}p^3 \leq \frac{n^3 p^3}{6}$$
and every rectangle accounts for three pairs of $(i,j)$ which will not be filtered out.\\

\textit{Case 2.} The edge $x \sim y$ has been added but $x \sim z$ and $z \sim y$ did already exist. 

This implies that $y$ is at distance 2
from $x$. The number of random connections emanating from $x$ is given by a binomial distribution with distribution
$$ k \sim \mathcal{B}(n-1-\delta(x),p), ~\mbox{where}~\mathcal{B}~\mbox{is the binomial distribution}$$
and $\delta(x)$ is the degree of $x$ in the original graph.
For every such new edge the probability of connecting to one of the less than $c$ vertices at distance at most 2 from $x$ is at most $c/(n-1)$ and thus
assuming the worst case $(\delta(x) = 0)$ and using the reproducing property of the binomial distribution, we are able to simpify the arising expression 
$$ \mathcal{B}\left(\mathcal{B}(n-1,p),\frac{c}{n-1}\right)= \mathcal{B}\left(n-1, \frac{pc}{n-1}\right).$$
This computation is for a fixed vertex $x$, using the linearity of expectation bounds the number of errors introduced through Case 2 by
$$ n\cdot \mathbb{E} \mathcal{B}\left(n-1, \frac{pc}{n-1}\right) = npc.$$

\textit{Case 3.} The edge $x \sim y$ and one of the two edges  $x \sim z$ or $z \sim y$ has been added while the other one did already exist.
We change our point of view. Fixing a vertex $z$, what can be said about the probability that there is another vertex $x$ such that both
edges $z \sim x$ and $x \sim y$ have been added for some $y$ at distance 1 from $x$?

\begin{center}
\begin{figure}[h!]
\begin{tikzpicture}[scale=1.3]
\draw [fill, ultra thick] (0,1) circle [radius=0.03];
\draw [fill, ultra thick] (1,0) circle [radius=0.03];
\draw [fill, ultra thick] (1,0.5) circle [radius=0.03];
\draw [fill, ultra thick] (1,1) circle [radius=0.03];
\draw [fill, ultra thick] (1,1.5) circle [radius=0.03];
\draw [fill, ultra thick] (1,2) circle [radius=0.03];

  \draw[dashed, thick] (0,1) --(1, 0);
  \draw[dashed, thick] (0,1) --(1, 0.5);
  \draw[dashed, thick] (0,1) --(1, 1);
  \draw[dashed, thick] (0,1) --(1, 1.5);
  \draw[dashed, thick] (0,1) --(1, 2);

  \draw[dashed, thick] (2,2.3) --(1, 2);
  \draw[dashed, thick] (2,2.1) --(1, 2);
  \draw[dashed, thick] (2,1.9) --(1, 2);

  \draw[dashed, thick] (2,1.8) --(1, 1.5);
  \draw[dashed, thick] (2,1.6) --(1, 1.5);
  \draw[dashed, thick] (2,1.4) --(1, 1.5);

  \draw[dashed, thick] (2,1.2) --(1, 1);
  \draw[dashed, thick] (2,1) --(1, 1);
  \draw[dashed, thick] (2,0.8) --(1, 1);

  \draw[dashed, thick] (2,0.6) --(1, 0.5);
  \draw[dashed, thick] (2,0.4) --(1, 0.5);
  \draw[dashed, thick] (2,0.2) --(1, 0.5);

  \draw[dashed, thick] (2,0.1) --(1, 0);
  \draw[dashed, thick] (2,-0.1) --(1, 0);
  \draw[dashed, thick] (2,-0.3) --(1, 0);

 \node[left] at (0,1) {$x$};
\draw (2,1) ellipse (0.5cm and 1.5cm);
\draw (2.2, 1) -- (3,0.5);
 \node[right] at (3,0.5) {any of these vertices close to $x$?};

\end{tikzpicture}
\caption{There is a random number of edges connecting to vertices from each of which there is another
random number of edges emanating -- what's the likelihood of a 2-neighbourhood of $x$ being reached in two steps?}
\end{figure}
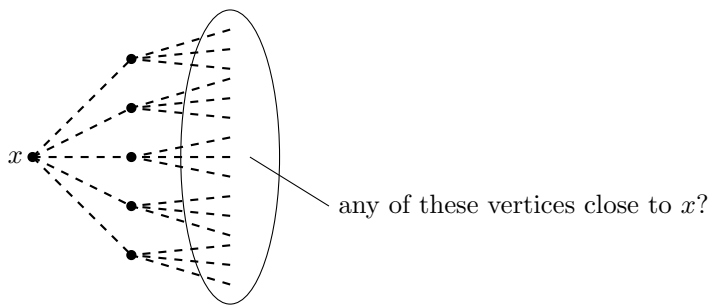
\end{center}

Let now $z$ be fixed
and assume that the randomly added edges adjacent to $v$ are $(z, z_1), \dots, (z,z_k)$ have been added. Note that again
$$ k \sim \mathcal{B}(n-1,p), ~\mbox{where}~\mathcal{B}~\mbox{is the binomial distribution}.$$
Let us fix $(z,z_1)$. We are then interested in the probability that another randomly added edge is connecting $z_1$
and one of the $c$ points at distance at most 2 from $z$. The number of added eges adjacent to $z_1$ is again $\mathcal{B}(n-1, p).$
Using again the reproducing
property of the binomial distribution as above, we are able to simpify the arising expression 
$$ \mathcal{B}\left(\mathcal{B}(n-1,p),\frac{c}{n-1}\right)= \mathcal{B}\left(n-1, \frac{pc}{n-1}\right).$$
This describes the distribution of numbers of tuples $z \rightarrow z_1 \rightarrow y$, where $y$ is at distance at most 2 from $x$. 
However, we are not limited to constructing such a path over $z_1$ since we have many $z_i$ at our disposal.
Indeed, there are  $z$ possible vertices, which implies that the total number of edges $z \rightarrow z_i \rightarrow y$, where $y$ is at distance at most 2 from $z$ is distributed like
a random sum of random variables
$$ X =  \underbrace{\mathcal{B}\left(n-1, \frac{pc}{n-1}\right) + \mathcal{B}\left(n-1, \frac{pc}{n-1}\right) + \dots + \mathcal{B}\left(n-1, \frac{pc}{n-1}\right)}_{ \mathcal{B}(n-1,p)}.$$
A basic result due to Wald \cite{wald1, wald2} is as follows: given a random number $N$ of i.i.d. variables $X$, their sum $S$ satisfies
$$ \mathbb{E} S = \left(\mathbb{E} N\right)\left(\mathbb{E} X\right).$$
In particular, this immediately implies that
$$ \mathbb{E}X = \left( \mathbb{E} \mathcal{B}\left(n-1, \frac{pc}{n-1}\right) \right)\left(\mathbb{E} \mathcal{B}(n-1,p) \right) = pc(n-1)p \leq c p^2 n.$$
Summing over all $n$ vertices gives the result.
\end{proof}

\section{Examples}
\subsection{Setup.} We study a series of classification problems. Since we are interested in the effectiveness of the method and not in the effectiveness of various spectral 
methods, we will fix one of the simplest spectral methods by defining the quadratic form to be
$$ \sum_{i=1}^{n}{t_{ij}f(x_i)f(x_j)} \qquad \mbox{where} \quad t_{ij} = \begin{cases}
-\sum_{k=1 \atop k \neq i}^n{s_{ik}} \qquad &\mbox{if}~i=j \\
s_{ij} \qquad &\mbox{if}~i\neq j, \\
\end{cases}$$
where $s_{ij}$ will be either $p_{ij}$ or $q_{ij}$. The Markov chain $P$ will be created using the classical isotropic Gaussian kernel with a parameter $\varepsilon$ (the precise value will
be explicitely noted in each application)
$$ p_{ij} = \frac{ \exp\left(-\frac{1}{\varepsilon} \|x_i -x_j\|^2_{\ell^2}\right)}{\sum_{k = 1 \atop k \neq j}^{n}{ \exp\left(-\frac{1}{\varepsilon} \|x_k -x_j\|^2_{\ell^2}\right)}}.$$
The matrix $Q$ will be created as above (with $k=2$ for classification examples and $k=5$ for error correction).
As for the embedding, we will always chose the simplest possible one into two-dimensional space
\begin{align*}
\Phi: \left\{x_i\right\}_{i=1}^{n} &\rightarrow \mathbb{R}^2 \\
x &\rightarrow  (\phi_1(x), \phi_2(x)),
\end{align*}
where $\phi_1, \phi_2$ are the first two eigenfunctions.
Each classification problems will be given as a point cloud $(x_{i})_{i=1}^{n} \subset \mathbb{R}^d$
and we shall assume that each coordinate has equal importance: this is done by dividing every coordinate by the standard deviation of that coordinate among all $n$ points.
 We emphasize that we are not at all interested in obtaining competitive result  -- we study the effect of the method (better results could, for example, be immediately obtained
by considering more than two eigenfunctions for an embedding). All data sets have been taken from the UCI machine learning depository.

\subsection{Wine dataset} The challenge is to automatically detect the origin of 178 wines
(each of which is given by the quantities of 13 chemical constitutents) grown by three different cultivars in the same region in Italy. The classical approach (using $\varepsilon = 10$) is very successful
at this task, the preconditioning matrix has little effect. It should be noted that the preconditioner successfully contracts the red and the green region (which is yet again accomplished by regularizing the
diffusion distance within those regions).
\begin{figure}[h!]
\begin{minipage}[b]{0.48\linewidth}
\centering
\includegraphics[width=\textwidth]{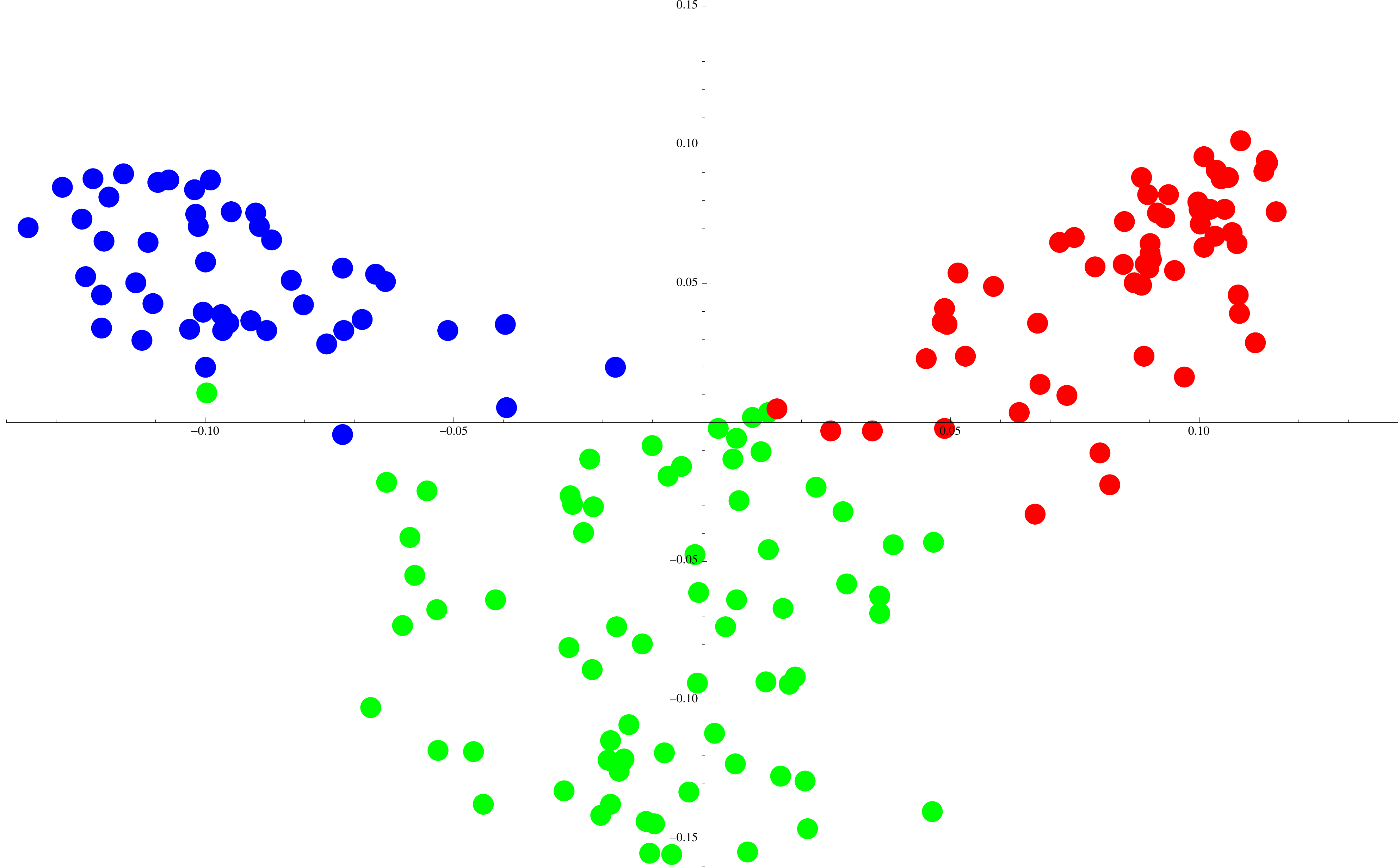}
\caption{Embedding of $P$}
\end{minipage}
\begin{minipage}[b]{0.48\linewidth}
\centering
\includegraphics[width=\textwidth]{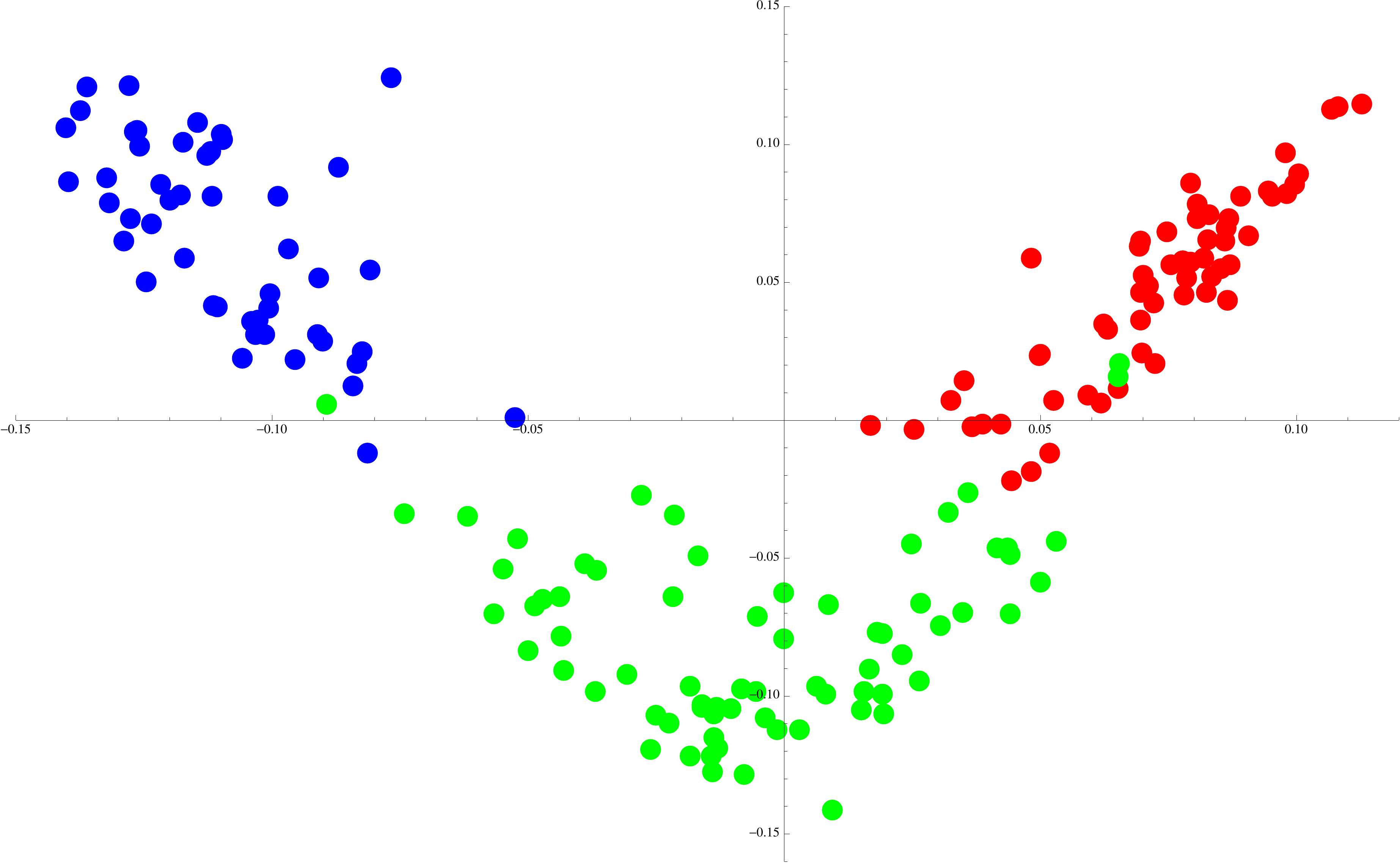}
\caption{Embedding of $Q$}
\end{minipage}
\end{figure}

\subsection{Cleveland Heart Disease dataset}  One is given 303 data samples, each of which gives 13 attributes (age, cholesterol, \dots). The task is to classify the
existence or absence of heart disease. The classical embedding (using $\varepsilon = 10$) manages some separation, the matrix $Q$ decreases the diffusion distance within
the subset of healthy patients (green) which allows for a more uniform embedding. The error rate (using a half space as classifier in both cases) is unaffected.
\begin{figure}[h!]
\begin{minipage}[b]{0.48\linewidth}
\centering
\includegraphics[width=\textwidth]{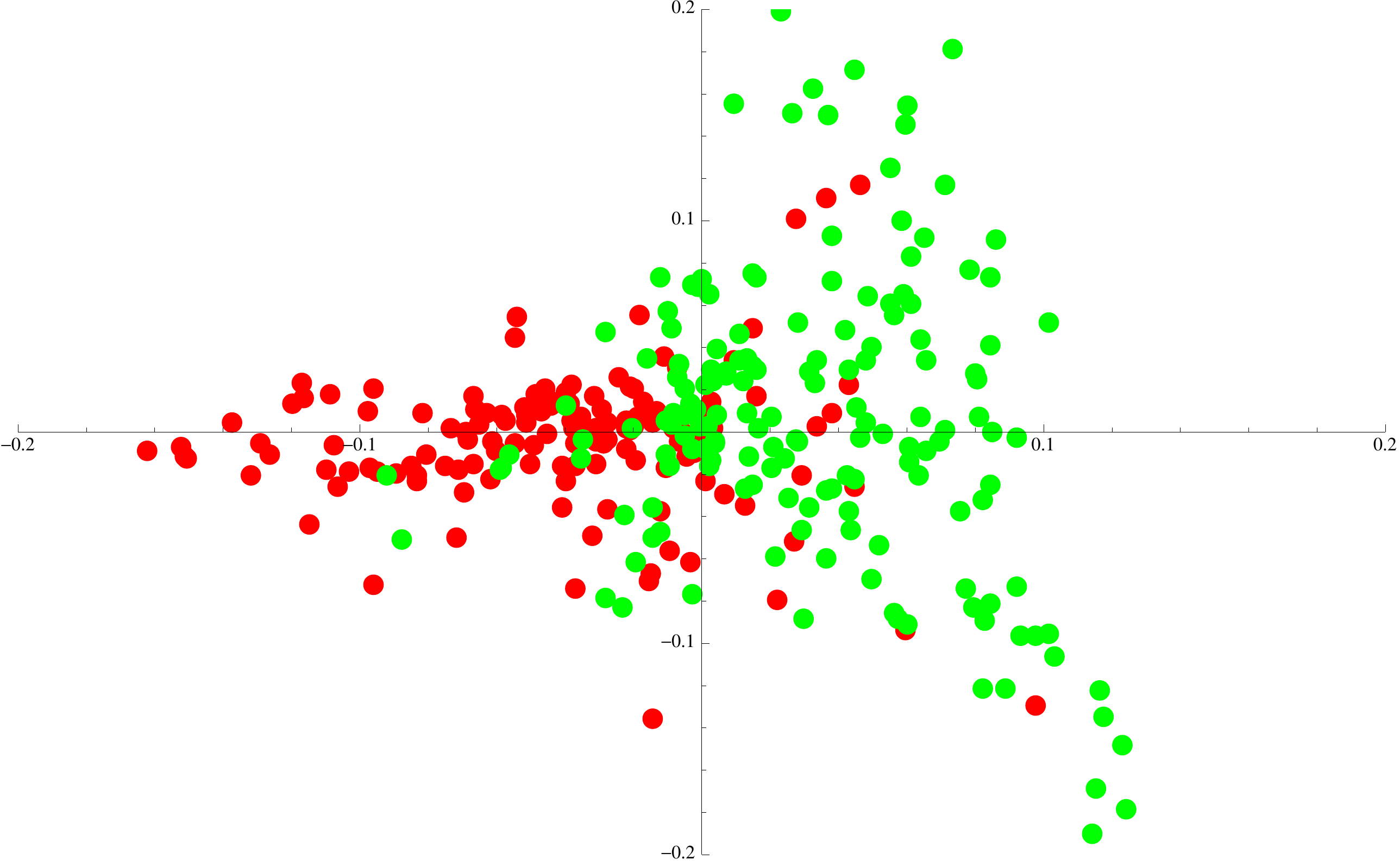}
\caption{Embedding of $P$}
\end{minipage}
\begin{minipage}[b]{0.48\linewidth}
\centering
\includegraphics[width=\textwidth]{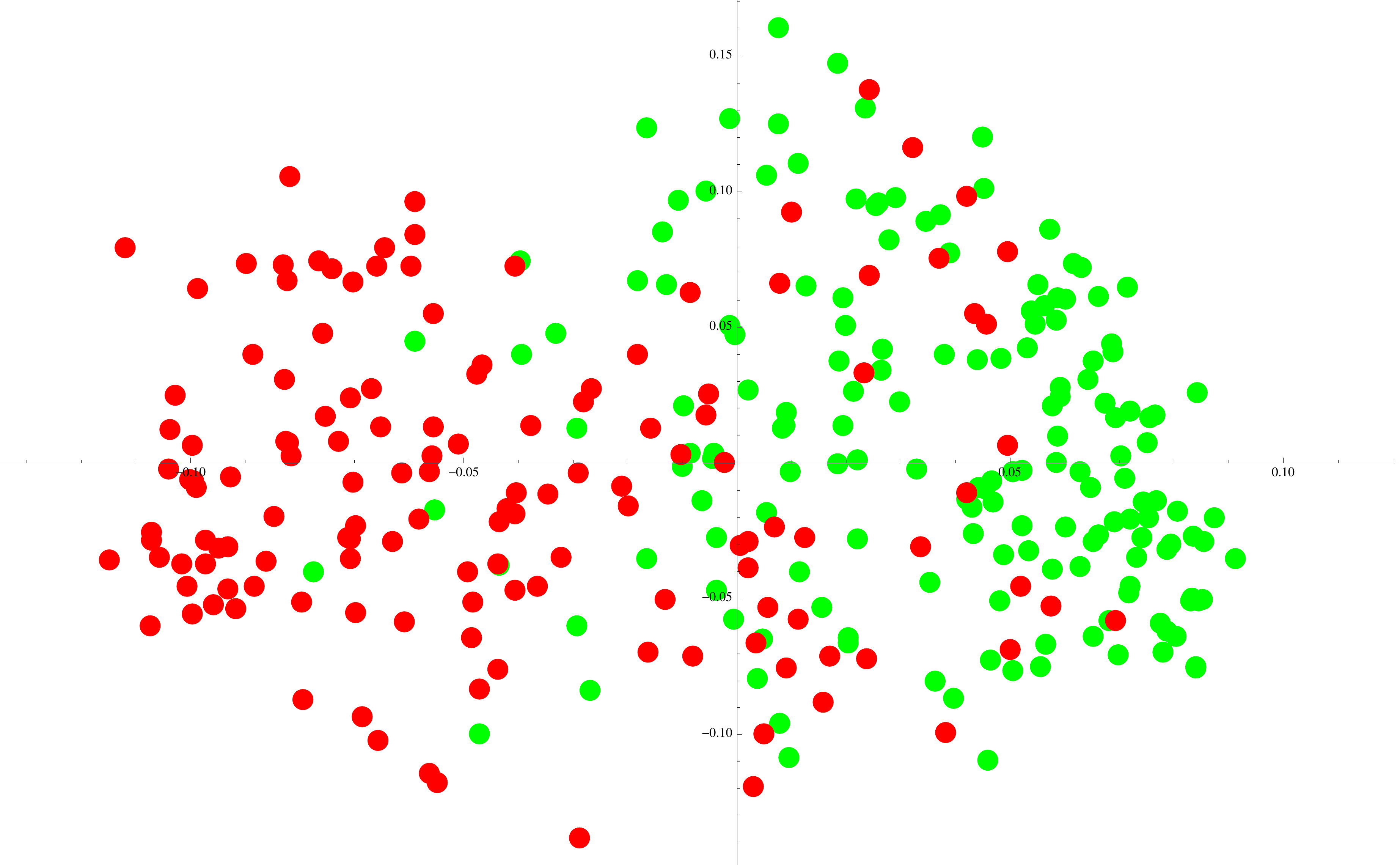}
\caption{Embedding of $Q$}
\end{minipage}
\end{figure}

\subsection{Wisconsin Diagnostic Breast Cancer (WDBC)} The data set contains 569 elements, each of which quantifies 10 aspects of a tumor  --  the task is to identify a tumor
as benign or malign. This is a data set where the classical embedding already works very well ($\varepsilon = 100$): using a half space as classifier, the error rate is a mere 5.9\% and the picture shows a clear separation into two cluster. The preconditioner achieves
a slightly cleaner separation and reduces to error (being again allowed one half space as a classifier) to 4.2\%.

\begin{figure}[h!]
\begin{minipage}[b]{0.48\linewidth}
\centering
\includegraphics[width=\textwidth]{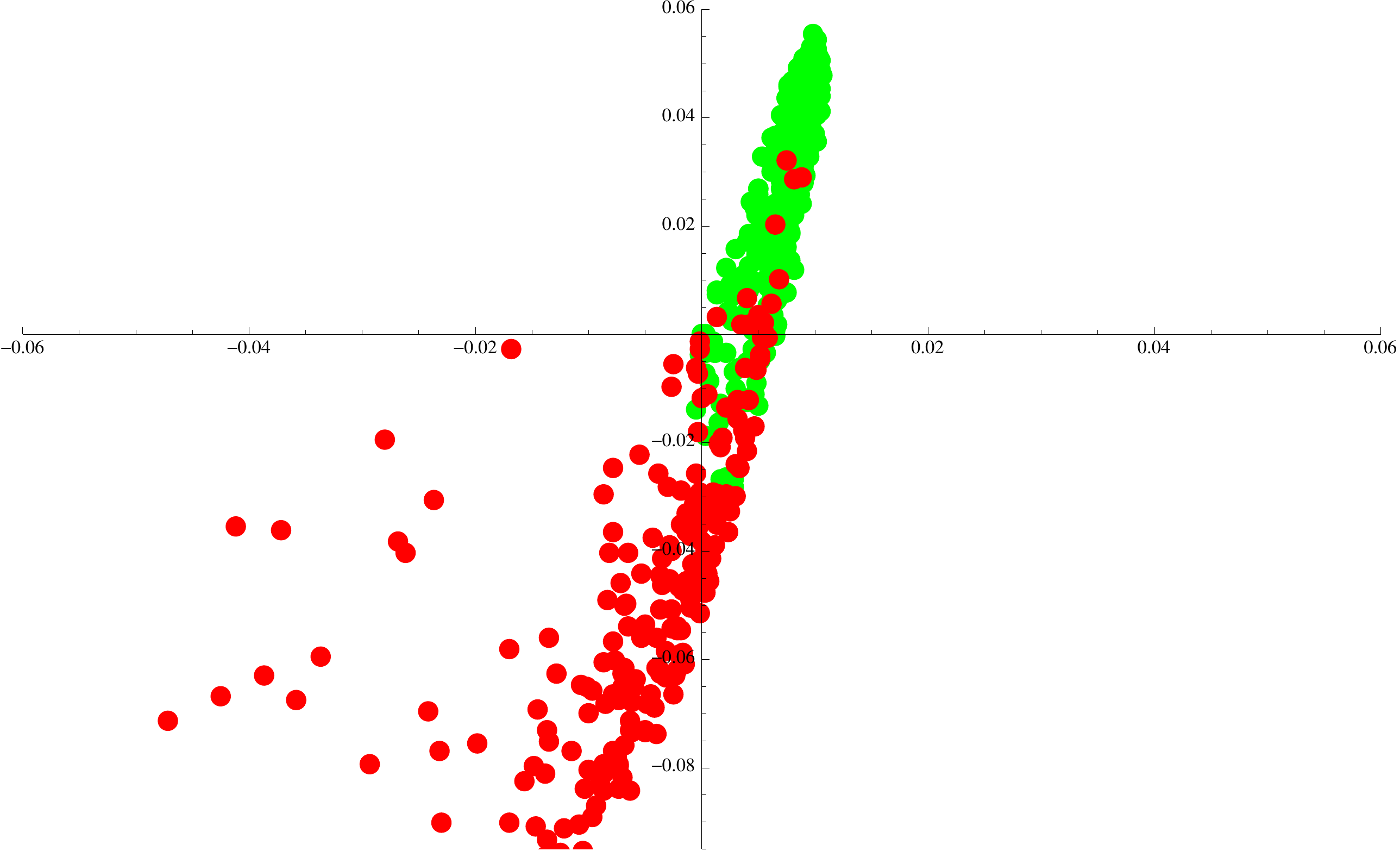}
\caption{Embedding of $P$}
\end{minipage}
\begin{minipage}[b]{0.48\linewidth}
\centering
\includegraphics[width=\textwidth]{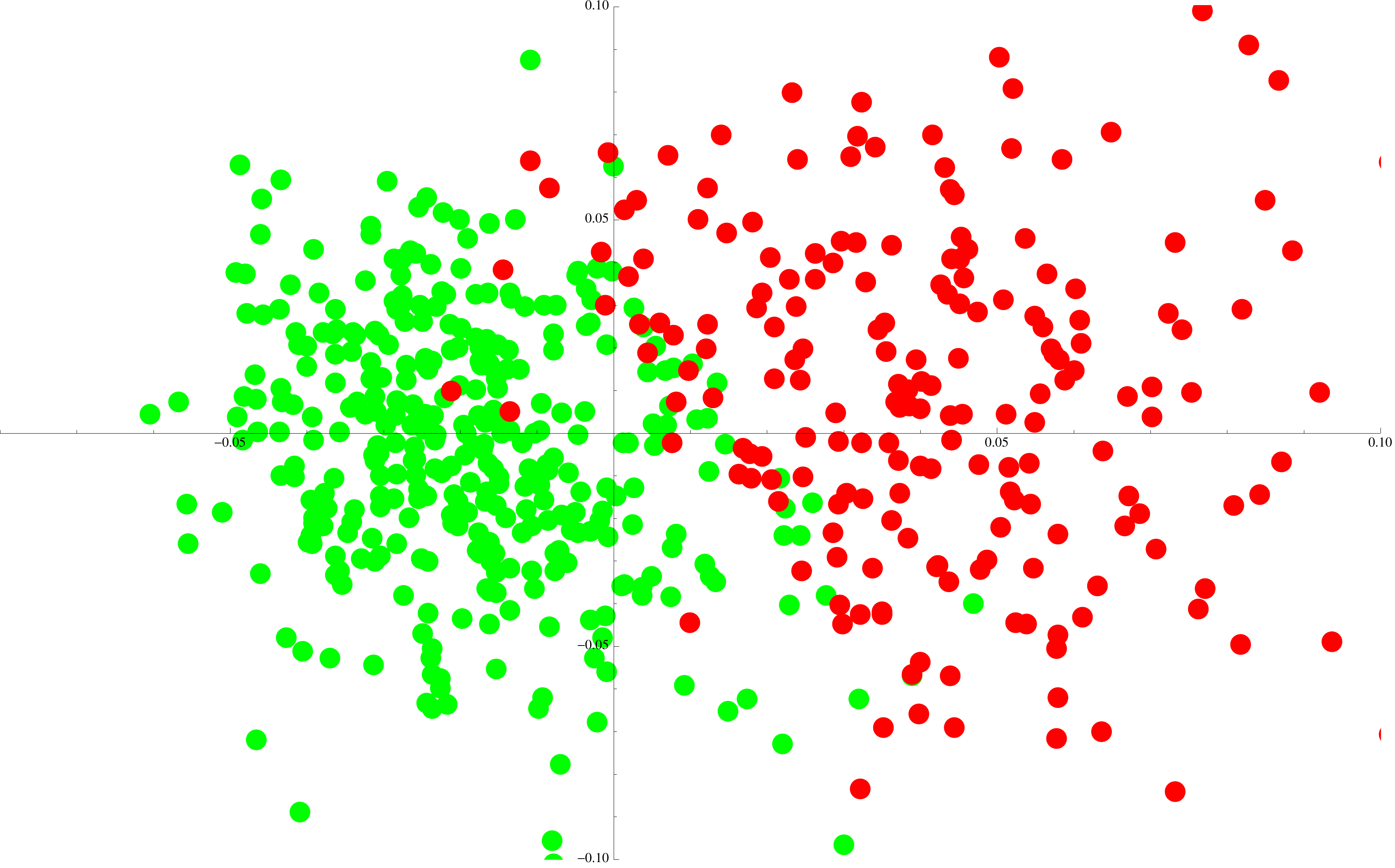}
\caption{Embedding of $Q$}
\end{minipage}
\end{figure}

\subsection{3D example: Ionosphere data set} This is another standard binary classification task. An array of antennas is trying to find structure in the ionosphere and there are 34 possible settings; those settings, which are able to detect structure are 'good' whereas other settings are 'bad'. The dataset is comprised of 351 elements. We observe again the same behavior as above; the unmodified embedding has
a smooth transition between the two points: the method (here with $k=5$) manages to contract the sets of good and bad points, respectively, and increase their distance. Note that the method manages to
contract almost all red outliers! The pictures shown use the canonical 3D embedding
$$ x \rightarrow (\phi_1(x), \phi_2(x), \phi_3(x)).$$ 
\begin{figure}[h!]
\begin{minipage}[b]{0.48\linewidth}
\centering
\includegraphics[width=\textwidth]{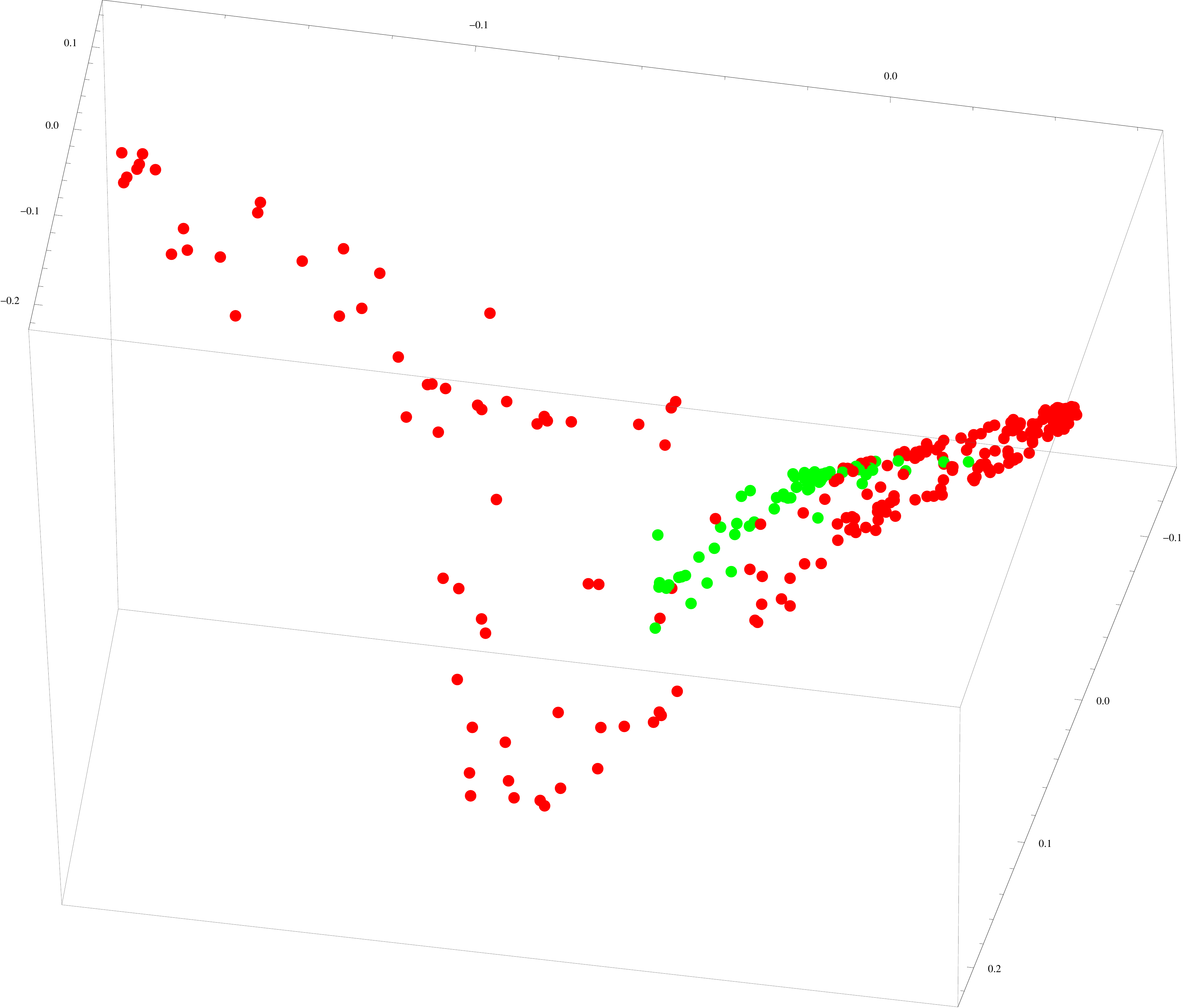}
\caption{Embedding of $P$}
\end{minipage}
\begin{minipage}[b]{0.48\linewidth}
\centering
\includegraphics[width=\textwidth]{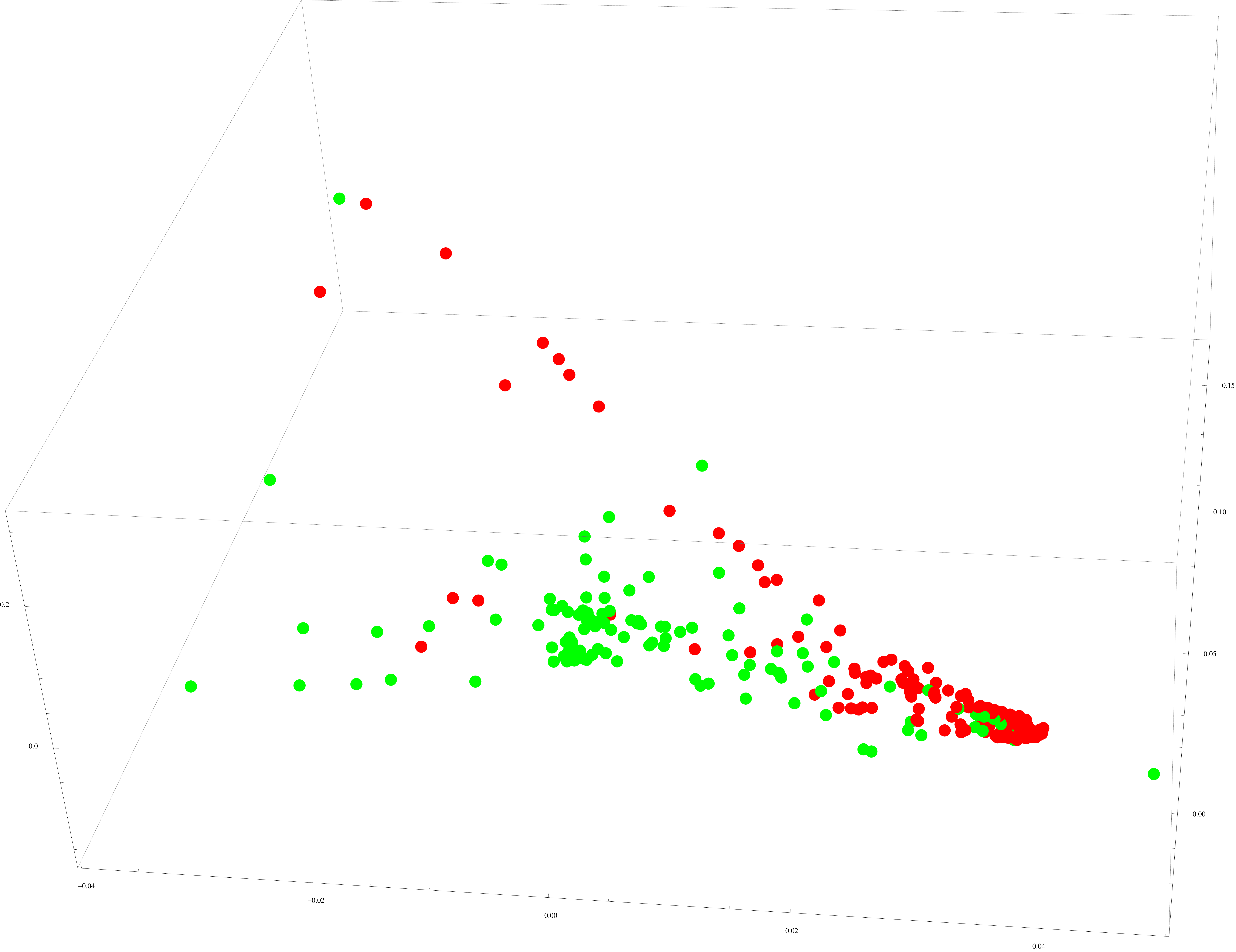}
\caption{Embedding of $Q$}
\end{minipage}
\end{figure}

\subsection{Correcting random errors} Finally, we consider the method with respect to correcting random errors. Consider the set of points in the shape of a circle
$$ \left\{ \left( \cos{\frac{2 \pi j}{100}}, \sin{\frac{2\pi j}{100}} \right): 1 \leq j \leq 100 \right\} \subset \mathbb{R}^2.$$
Spectral methods perform very well and are able to reproduce the underlying shape.
We assume that 
$$ p_{ij} = \frac{ \exp\left(-50 \|x_i -x_j\|^2_{\ell^2}\right)}{\sum_{k = 1 \atop k \neq j}^{n}{ \exp\left(-50 \|x_k -x_j\|^2_{\ell^2}\right)}} $$
and randomize the adjacency matrix by picking $k$ random edges and defining the weight of that edge to be 1. \\

The following example show the same randomly perturbed set as embedded by $P$ and $Q$ (with $k=5$).
 $Q$ is able to maintain the topological structure rather well.
\begin{figure}[h!]
\begin{minipage}[b]{0.47\linewidth}
\centering
\includegraphics[width=\textwidth]{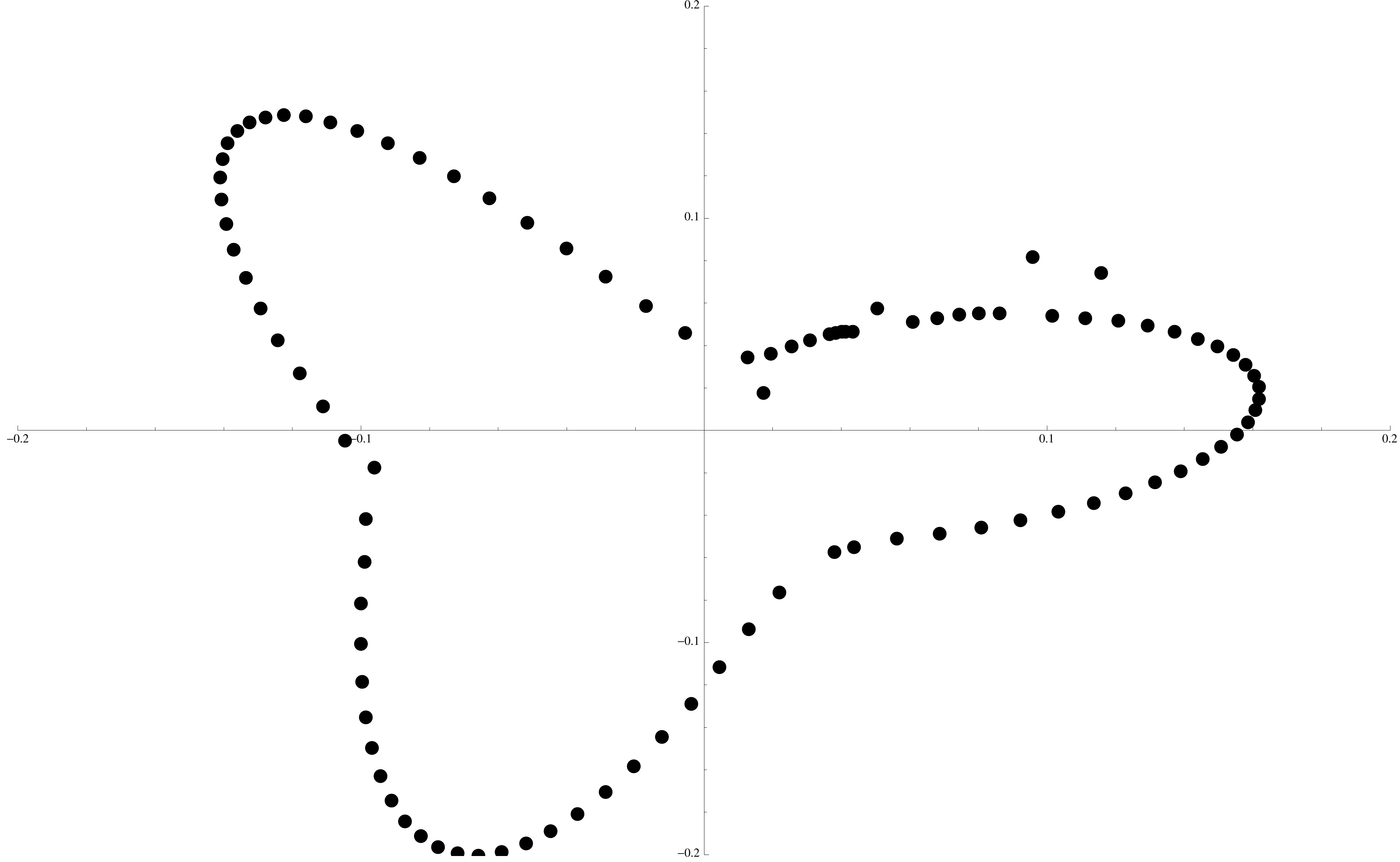}
\caption{Embedding of $P$ (5 random edges).}
\end{minipage}
\hspace{0.1cm}
\begin{minipage}[b]{0.47\linewidth}
\centering
\includegraphics[width=\textwidth]{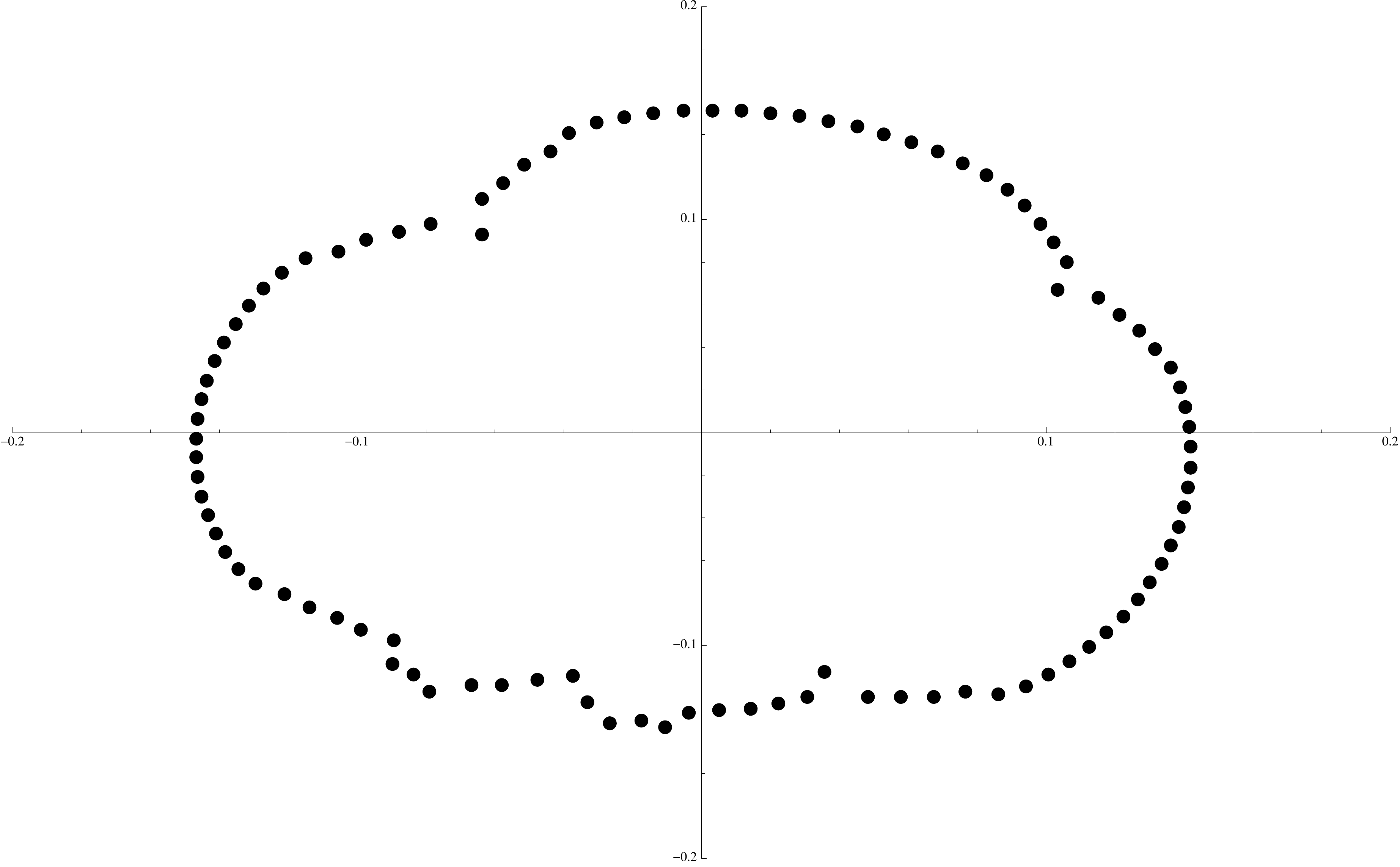}
\caption{Embedding of $Q$ (5 random edges, $k=5$).}
\end{minipage}
\end{figure}

\begin{figure}[h!]
\begin{minipage}[b]{0.47\linewidth}
\centering
\includegraphics[width=\textwidth]{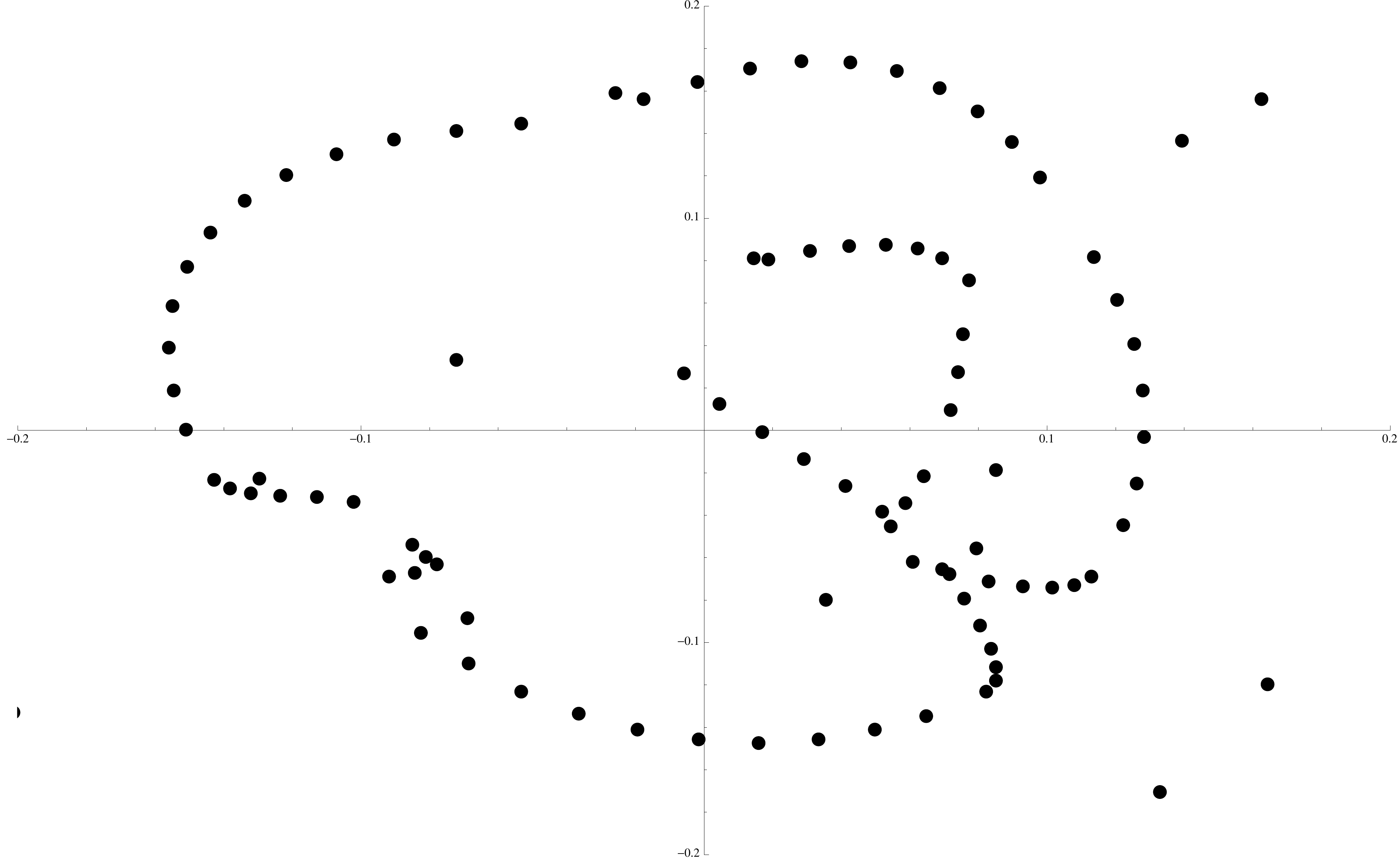}
\caption{Embedding of $P$ (10 random edges).}
\end{minipage}
\hspace{0.1cm}
\begin{minipage}[b]{0.47\linewidth}
\centering
\includegraphics[width=\textwidth]{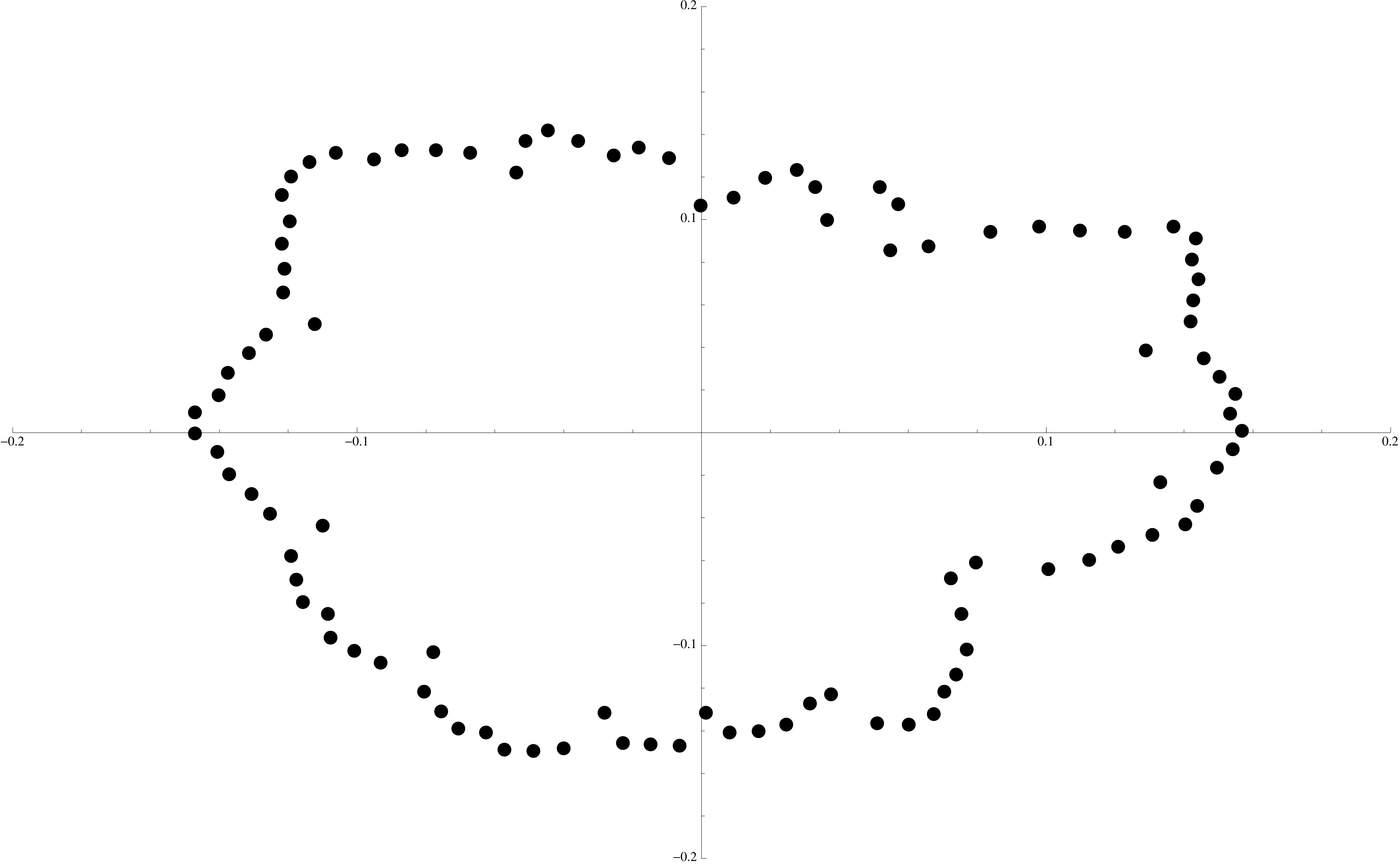}
\caption{Embedding of $Q$ (10 random edges, $k=5$).}
\end{minipage}
\end{figure}

\begin{figure}[h!]
\begin{minipage}[b]{0.47\linewidth}
\centering
\includegraphics[width=\textwidth]{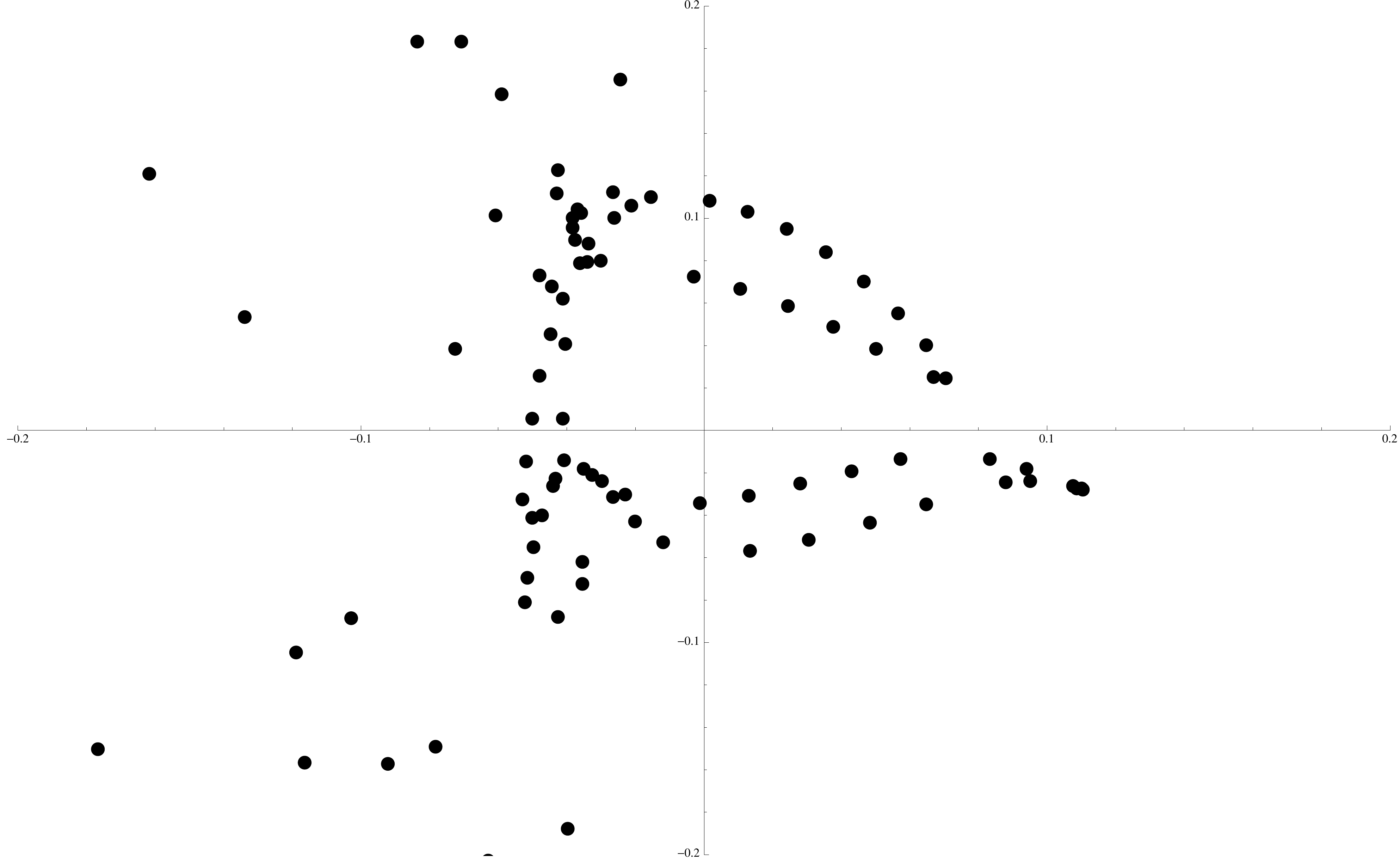}
\caption{Embedding of $P$ (20 random edges).}
\end{minipage}
\hspace{0.1cm}
\begin{minipage}[b]{0.47\linewidth}
\centering
\includegraphics[width=\textwidth]{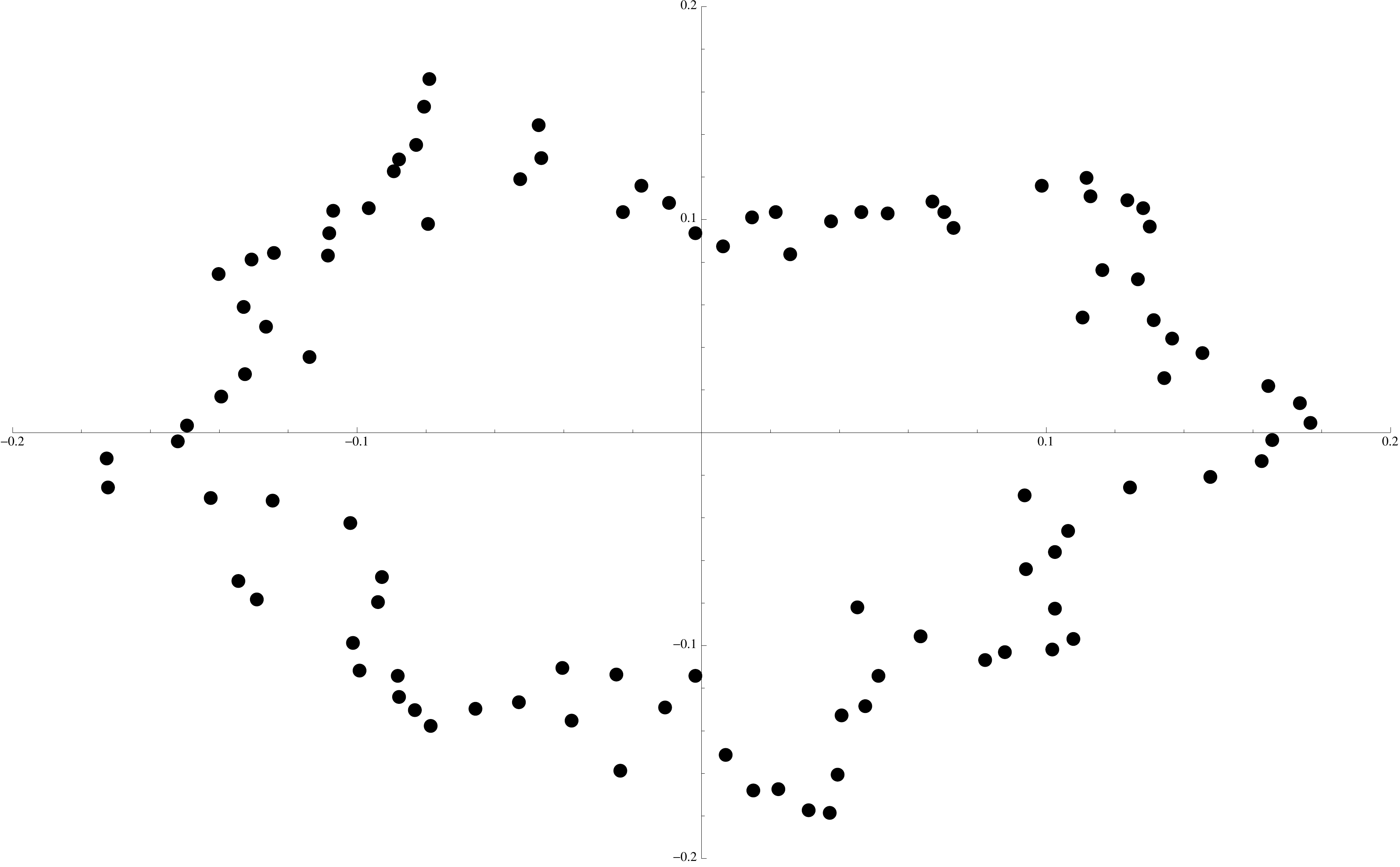}
\caption{Embedding of $Q$ (20 random edges, $k=5$).}
\end{minipage}
\end{figure}

\subsection{Concluding Remark.} We remark that our method is not the unique one with that property; the theorem about error correcting also holds, for example, for the pointwise product
$$ q_{ij} \sim P^*_{ij} \cdot ((P^*)^2)_{ij}.$$
It is certainly conceivable that another implementation of the underlying idea of exploiting local topology could
yield even better results and we consider this to be a very interesting problem.\\

\textbf{Acknowledgement.} The author is indebted to Raphy Coifman for extensive discussions.

\end{document}